\documentclass[preprint,showpacs,amsmath]{revtex4}
\usepackage{amssymb}
\usepackage{graphicx,epsfig,subfigure,dsfont,amssymb,amsthm,amsfonts,amsbsy,mathrsfs,amscd}
\def\qed{\leavevmode\unskip\penalty9999 \hbox{}\nobreak\hfill
     \quad\hbox{\leavevmode  \hbox to.77778em{%
               \hfil\vrule   \vbox to.675em%
               {\hrule width.6em\vfil\hrule}\vrule\hfil}}
     \par\vskip3pt}

\newtheorem{theorem}{Theorem}

\newtheorem{lemma}{Lemma}
\newtheorem{definition}{Definition}
\input amssym.def
\begin{document}

\begin{center}
\bf{Coherence Concurrence for X States}
\end{center}

\begin{center}
{Ming-Jing Zhao$^{1}$, Teng Ma$^{2,3,4}$, Zhen Wang$^5$, Shao-Ming Fei$^{6,7}$, Rajesh Pereira$^8$}
\end{center}

\begin{center}
\begin{minipage}{5in}
\small $~^{1}$ {School of Science, Beijing
Information Science and Technology University, 100192, Beijing,
China}

\small $~^{2}$ {Shenzhen Institute for Quantum Science and Engineering and Department of Physics, Southern University of Science and Technology, Shenzhen 518055, China}

\small $~^{3}$ {Key Laboratory of Quantum Information,
University of Science and Technology of China, CAS, Hefei 230026, China}

\small $~^{4}$ {Shenzhen Key Laboratory of Quantum Science and Engineering, Southern University of Science and Technology, Shenzhen 518055, China}

\small $~^{5}$ {Department of Mathematics, Jining University,
Qufu 273155, China}

\small $~^{6}$ {School of Mathematical Sciences, Capital Normal
University, Beijing
100048, China}

\small $~^{7}$ {Max-Planck-Institute for Mathematics in the Sciences, 04103
Leipzig, Germany}

\small $~^{8}$ {Department of Mathematics and Statistics, University of Guelph, N1G2W1, Canada}

\end{minipage}
\end{center}


\textbf{Abstract}\ \
We study the properties of coherence concurrence and present a physical explanation analogous to the coherence of assistance. We give an optimal pure state decomposition which attains the coherence concurrence
for qubit states. We prove the additivity of coherence concurrence under direct sum operations in another way. Using these results, we calculate analytically the coherence concurrence for X states and show its optimal decompositions. Moreover, we show that the coherence concurrence is exactly twice the convex roof extended negativity of the Schmidt correlated states, thus establishing a direct relation between coherence concurrence and quantum entanglement.


\section{Introduction}

Quantum coherence is an important feature in quantum physics and is of practical significance in quantum computation and quantum communication \cite{A. Streltsov-rev,E. Chitambar-review,M. Hu}.
The formulation of the resource theory of coherence was initiated in Ref. \cite{T. Baumgratz}, in which some intuitive and
computable measures of coherence are identified, for example, the $l_1$-norm coherence and the relative entropy coherence. These coherence measures quantify coherence by using the minimal distance between the quantum state and the set of incoherent states.
Operationally, distillable coherence and coherence cost are two quantum measures quantifying the optimal rate in transformation between quantum states and maximally coherent states under incoherent operations in the asymptotic limit \cite{A. Winter}.
Robustness of coherence is a coherence monotone which quantifies the minimal mixing required to make a state incoherent \cite{C. Napoli}, from which the witness observable has been demonstrated \cite{W. Zheng}. The skew information based coherence has been proposed as a characterization of the uncertainty of the system being measured \cite{D. Girolami}.

If a coherence measure is defined for all pure states, it can be extended to all mixed states using the convex roof construction. For instance, the
intrinsic randomness of coherence and coherence of information are convex roof extended coherence measures based on the relative entropy coherence \cite{A. Winter,X. Yuan}, while the coherence concurrence is based on $l_1$-norm coherence \cite{X. Qi}. The coherence number is also a convex roof extended discrete coherence monotone based on the Schmidt numbers \cite{S. Chin}, as is the fidelity-based measure of coherence \cite{C. Liu}.

Although the convex roof extended coherence quantifiers are valid coherence measures provided they are valid for pure states, they are not easy to calculate in general since these calculation involves minimizations. For single qubit states, many convex roof extended coherence measures including intrinsic randomness, coherence concurrence, fidelity-based measure of coherence have analytical expressions \cite{X. Yuan,X. Qi,C. Liu}. But the situation becomes much more complicated for three or higher dimensional systems.

In this paper, we focus on the coherence concurrence which is the convex roof extension of the $l_1$-norm coherence on pure states. We first analyze the optimal pure state decomposition for single qubit states. Then we show the additivity of coherence concurrence under the direct sum operations in another way. Based on the additivity we calculate analytically the coherence concurrence for X states and show its optimal pure state decomposition. Finally, we present the relation between the coherence concurrence and entanglement.

\section{The Coherence Concurrence for X states}

Under a fixed reference basis $\{|i\rangle\}$,
a quantum state $\rho$ is said to be incoherent if the state is diagonal in this basis, i.e. $\rho=\sum \rho_{i}|i\rangle \langle i|$.
Otherwise the quantum state is said to be coherent. One commonly used coherence measure is called the $l_1$-norm coherence.
\begin{definition}
The $l_1$-norm coherence of a quantum state $\rho=\sum \rho_{ij}|i\rangle \langle j|$ is the sum of the magnitudes of all off-diagonal entries, $C_{l_1}(\rho)=\sum_{i\neq j} |\rho_{ij}|$.
\end{definition}
Based on the $l_1$-norm coherence, the coherence concurrence of $\rho$ is proposed in the convex roof
construction \cite{X. Qi}.
\begin{definition}
The coherence concurrence $C_{l_1}^{c}$ of $\rho$ is
\begin{equation}
C_{l_1}^{c}(\rho)=\min \sum_i p_i C_{l_1}(|\psi_i\rangle\langle \psi_i|),
\end{equation}
where the minimization is taken over all pure state decompositions of $\rho=\sum_i p_i |\psi_i\rangle\langle\psi_i|$.
\end{definition}

Dual to the definition of coherence concurrence,
the $l_1$-norm coherence of assistance is the maximal average $l_1$-norm coherence
$C_{l_1}^{a}(\rho)=\max \sum_k p_k C_{l_1}(|\psi_k\rangle\langle \psi_k|)$,
where the maximization is taken over all pure state decompositions of $\rho=\sum_k p_k |\psi_k\rangle\langle\psi_k|$ \cite{M. J. Zhao-l1-ca}.
Employing the physical illustration of the coherence of assistance, the coherence concurrence $C_{l_1}^{c}$ has the following operational interpretation.
Suppose Alice holds a state $\rho^A$. Bob holds another part of the purified state of $\rho^A$. Bob performs local measurements and informs Alice of the measurement outcomes by classical communication. Alice's quantum state will be in one pure state ensemble $\{ p_k,\  |\psi_k\rangle\langle \psi_k|\}$ with average $l_1$-norm coherence $\sum_k p_k C_{l_1}(|\psi_k\rangle\langle \psi_k|)$. As $l_1$-norm coherence is a convex function, the $l_1$-norm coherence can be increased minimally to $C_{l_1}^{c}(\rho^A)$ by such process.

For a two dimensional quantum state $\rho=
\left(\begin{array}{cccccccc}
\rho_{11} &\rho_{12}\\
\rho_{12}^*&\rho_{22}
\end{array}\right)$, the coherence concurrence has been shown to be the $l_1$-norm coherence $C_{l_1}^{c}(\rho)=C_{l_1}(\rho)=2|\rho_{12}|$ \cite{X. Yuan,X. Qi}. We first present an optimal pure state decomposition attaining the minimum average $l_1$-norm coherence for the qubit state $\rho$.
We assume $0<\rho_{11}\leq \rho_{22}<1$. Let $|\psi_1\rangle\langle\psi_1|=\frac{1}{p_1}\left(\begin{array}{cccccccc}
\rho_{11} &\rho_{12}\\
\rho_{12}^*& |\rho_{12}|^2/\rho_{11}
\end{array}\right)$ and $|\psi_2\rangle\langle\psi_2|=|1\rangle \langle1|$ with $p_1=\rho_{11}+|\rho_{12}|^2/\rho_{11}$ and $p_2=\rho_{22}-|\rho_{12}|^2/\rho_{11}$. Then $\{p_i,\ |\psi_i\rangle\langle \psi_i|\}_{i=1}^2$ is a pure state decomposition with
the average $l_1$-norm coherence being the same as the $l_1$-norm coherence $2|\rho_{12}|$.

However, the coherence concurrence and $l_1$-norm coherence are not necessarily equal in higher dimensional systems.
For example, consider the three dimensional quantum state
$\rho_x=\frac{1}{3}
\left(\begin{array}{cccccccc}
1 & 0 & 1\\
0 & 1 & x\\
1 & x & 1
\end{array}\right)
$
with 0$<|x|\leq 1$.
For any pure state decomposition $\{p_k, |\psi_k\rangle\langle \psi_k|\}$ of $\rho_x$, one can check directly that there are at least two pure states $|\psi_{k_1}\rangle$ and $|\psi_{k_2}\rangle$ in the ensemble with three nonzero coefficients, $|\psi_{k_1}\rangle=\sum_{i=1}^3 a^{(k_1)}_i |i\rangle$ and $|\psi_{k_2}\rangle=\sum_{i=1}^3 a^{(k_2)}_i |i\rangle$ with  $a^{(k_1)}_i,\ a^{(k_2)}_i\neq 0$ for $i=1,2,3$. By the convexity of the $l_1$-norm coherence,
it is easy to show that
the average $l_1$-norm coherence is strictly larger than the $l_1$-norm coherence, namely, $C_{l_1}^{c}(\rho_x)>C_{l_1}(\rho_x)$.

Before calculating the coherence concurrence for X states, we show the additivity for coherence concurrence under direct sum operation first.
The strong monotonicity and convexity
of a coherence measure are in fact equivalent to the additivity of coherence for
subspace independent states \cite{X. D. Yu}. The coherence concurrence as a valid coherence measure
should satisfy the additivity under the direct sum operation. Here we give another proof of the additivity of coherence concurrence under the direct sum operation and explore its optimal pure state decompositions.

\begin{lemma} \cite{E. Sch} \label{lemma different decomposition}
Suppose $\rho=\sum_{l=1}^n \lambda_l |\psi_l\rangle \langle \psi_l|$ and $\rho=\sum_{k=1}^m p_k |\phi_k\rangle \langle \phi_k|$ are two arbitrary pure state decompositions of given quantum state $\rho$ with $\sum_{l=1}^n \lambda_l=\sum_{k=1}^m p_k=1$, $0\leq \lambda_l\leq 1$, $0\leq p_k\leq 1$ for $l=1,\cdots,n$, $k=1,\cdots,m$.
Then these two pure state decompositions are related by a transformation:
\begin{equation}\label{eq relation between two decomposition}
\sqrt{p_k}|\phi_k\rangle=\sum_{l=1}^n U_{lk} \sqrt{\lambda_l}|\psi_l\rangle,\ \ \ k=1,\cdots,m,
\end{equation}
where $U=(U_{lk})$ satisfying $UU^\dagger=I_{n\times n}$.
\end{lemma}
Here the transformation matrix $U$ is not necessarily square.
It should be also noted that the normalizer of a quantum state is not essential for the $l_1$-norm coherence as the $l_1$-norm is potentially homogenous. Hence we refer to the $l_1$-norm coherence of unnormalized density matrix sometimes for simplicity.

\begin{definition}
The direct sum of quantum states $\rho_i$ with probability $\sigma_i$, where $\sum_{i=1}^K \sigma_i=1$, $\sigma_i>0$ for $i=1,2,\cdots, K$, is the quantum state $\rho$ with density matrix in block diagonal form such as
$$\rho=\sigma_1\rho_1\oplus \sigma_2\rho_2\oplus \cdots\oplus \sigma_K\rho_K=\left(
\begin{array}{cccc}
\sigma_1\rho_1 & 0 & \cdots & 0 \\
0 &\sigma_2\rho_2 & \cdots & 0\\
\cdots & \cdots & \cdots & \cdots \\
0 & 0 & \cdots & \sigma_K\rho_K
\end{array}
\right).$$
\end{definition}

\begin{theorem}\label{th direct sum}
If a quantum state $\rho$ is the direct sum of some states $\rho_i$ with probability $\sigma_i$,
$\rho=\sigma_1\rho_1\oplus \sigma_2\rho_2\oplus \cdots\oplus \sigma_K\rho_K$, where $\sum_{i=1}^K \sigma_i=1$, $\sigma_i>0$ for $i=1,2,\cdots, K$,
then the coherence concurrence of $\rho$ is
\begin{eqnarray}\label{theorem1}
C_{l_1}^{c}(\rho)=\sigma_1C_{l_1}^{c}(\rho_1)+\sigma_2C_{l_1}^{c}(\rho_2)+\cdots+\sigma_KC_{l_1}^{c}(\rho_K).
\end{eqnarray}
The optimal pure state decomposition for $\rho$ attaining $C_{l_1}^{c}(\rho)$ is the union of the optimal pure state decompositions of $\sigma_i \rho_i$ attaining $C_{l_1}^{c}(\rho_i)$, $i=1,2,\cdots, K$.
\end{theorem}

Note that when we say $\{p_s,\  |\psi_s\rangle \langle \psi_s|\}$ is a pure state decomposition of an unnormalized quantum state $\rho$, we mean that $\{|\psi_s\rangle\}$ are normalized states and $\sum_s p_s=Tr(\rho)$.

\begin{proof}
Here we only need to prove the case that $C_{l_1}^{c}(\rho)=\sigma_1C_{l_1}^{c}(\rho_1)+\sigma_2C_{l_1}^{c}(\rho_2)$ for  $\rho=\sigma_1\rho_1\oplus \sigma_2\rho_2$.
Suppose $\{p_s,\  |\psi_s\rangle \langle \psi_s|\}_{s=1,\cdots,W}$
is a pure state decomposition attaining the minimum average $l_1$-norm coherence of $\sigma_1\rho_1$,
\begin{equation}
\sigma_1C_{l_1}^{c}(\rho_1)=
C_{l_1}^{c}(\sigma_1\rho_1)=\sum_{s=1}^W p_s C_{l_1}(|\psi_s\rangle\langle \psi_s|),
\end{equation}
with normalized pure state $|\psi_s\rangle =\sum_{i=1}^r a_i^{(s)}|i\rangle$, $s=1,\cdots, W$, and $\sum_{s=1}^W p_s=\sigma_1$; and $\{p_s,\  |\psi_s\rangle \langle \psi_s| \}_{s=W+1,\cdots,X}$ is a pure state decomposition attaining the minimum average $l_1$-norm coherence of $\sigma_2\rho_2$,
\begin{equation}
\sigma_2C_{l_1}^{c}(\rho_2)=C_{l_1}^{c}(\sigma_2\rho_2)=\sum_{s=W+1}^X p_s C_{l_1}(|\psi_s\rangle\langle \psi_s|),
\end{equation}
with normalized pure state $|\psi_s\rangle =\sum_{i=r+1}^n a_i^{(s)}|i\rangle$, $s=W+1,\cdots, X$, $X>W$, and $\sum_{s=1}^W p_s=\sigma_2$. Then  $\{p_s,\  |\psi_s\rangle\langle \psi_s| \}_{s=1,\cdots,W}\bigcup \{p_s,\  |\psi_s\rangle\langle \psi_s|\}_{s=W+1,\cdots,X}=\{p_s,\  |\psi_s\rangle\langle \psi_s|\}_{s=1,\cdots,X}$ is a pure state decomposition for $\rho$.
By definition, the minimum average $l_1$-norm coherence of $\rho$ satisfies
\begin{equation}\label{th eq cl1-c <}
\begin{array}{rcl}
C_{l_1}^{c}(\rho)&\leq& \sum_{s=1}^W p_s C_{l_1}(|\psi_s\rangle\langle \psi_s|) + \sum_{s=W+1}^X p_s C_{l_1}(|\psi_s\rangle\langle \psi_s|)\\
&=& \sigma_1C_{l_1}^{c}(\rho_1)+\sigma_2C_{l_1}^{c}(\rho_2).
\end{array}
\end{equation}

From Eq. (\ref{eq relation between two decomposition}), any other pure state decomposition $\{q_t,\  |\phi_t\rangle \langle\phi_t|\}_{t=1,\cdots,Y}$ of $\rho$ can be written as $\sqrt{q_t}|\phi_t\rangle=\sum_{s=1}^X U_{st} \sqrt{p_s}|\psi_s\rangle$, $t=1,\cdots,Y$.
We can partition each matrix $q_t|\phi_t\rangle\langle\phi_t|$ into four blocks, one diagonal block with the first $W$ rows and the first $W$ columns, $A_1^{(t)}=\sum_{s=1}^W \sum_{s^\prime=1}^W U_{st}U_{s^\prime t}^\dagger \sqrt{p_s}  \sqrt{p_{s^\prime}} |\psi_s\rangle \langle \psi_{s^\prime}|$; one off diagonal block with the first $W$ rows and the last $X-W$ columns, $A_2^{(t)}=\sum_{s=1}^W \sum_{s^\prime=W+1}^X U_{st}U_{s^\prime t}^\dagger \sqrt{p_s}  \sqrt{p_{s^\prime}} |\psi_s\rangle \langle \psi_{s^\prime}|$; one off diagonal block with the last $X-W$ rows and the first $W$ columns, $A_3^{(t)}=\sum_{s=W+1}^X \sum_{s^\prime=1}^W U_{st}U_{s^\prime t}^\dagger \sqrt{p_s}  \sqrt{p_{s^\prime}} |\psi_s\rangle \langle \psi_{s^\prime}|$; and the last diagonal block with the last $X-W$ rows and the last $X-W$ columns, $A_4^{(t)}=\sum_{s=W+1}^X \sum_{s^\prime=W+1}^X U_{st}U_{s^\prime t}^\dagger \sqrt{p_s}  \sqrt{p_{s^\prime}} |\psi_s\rangle \langle \psi_{s^\prime}|$, $t=1,2,\cdots, Y$.
That is
\begin{equation*}
q_t|\phi_t\rangle\langle\phi_t|=
\left(
\begin{array}{cc}
A_1^{(t)} & A_2^{(t)}\\
A_3^{(t)} & A_4^{(t)}
\end{array}
\right),\ \ t=1,2,\cdots, Y.
\end{equation*}
The $l_1$-norm coherence of $\rho$ comes from all off diagonal entries of the diagonal blocks $A_1^{(t)}$ and $A_4^{(t)}$ and all entries of the off diagonal blocks $A_2^{(t)}$ and $A_3^{(t)}$, $t=1,2,\cdots, Y$. Therefore,
\begin{equation}
\begin{array}{rcl}
\sum_{t=1}^Y q_t C_{l_1}(|\phi_t\rangle \langle\phi_t|)&=&\sum_{t=1}^Y C_{l_1}(\sum_{s=1}^X \sum_{s^\prime=1}^X U_{st} \sqrt{p_s} U_{s^\prime t}^\dagger \sqrt{p_{s^\prime}} |\psi_s\rangle \langle \psi_{s^\prime}|)\\
&\geq&\sum_{t=1}^Y C_{l_1}(\sum_{s=1}^W \sum_{s^\prime=1}^W U_{st} \sqrt{p_s} U_{s^\prime t}^\dagger \sqrt{p_{s^\prime}} |\psi_s\rangle \langle \psi_{s^\prime}|)\\&&+\sum_{t=1}^Y C_{l_1}(\sum_{s=W+1}^X \sum_{s^\prime=W+1}^X U_{st} \sqrt{p_s} U_{s^\prime t}^\dagger \sqrt{p_{s^\prime}} |\psi_s\rangle \langle \psi_{s^\prime}|),
\end{array}
\end{equation}
where we have gotten rid of the magnitudes of all entries of the off diagonal blocks $A_2^{(t)}$ and $A_3^{(t)}$ in the above inequality.

Similarly, we partition the matrix $U=(U_{st})$ into two blocks, one block with the first $W$ rows and the other block with the last $X-W$ rows as
$U=\left(
\begin{array}{cc}
U^{(1)}\\U^{(2)}
\end{array}
\right)$ with $U^{(1)}U^{(1)\dagger}=I_{W\times W}$ and $U^{(2)}U^{(2)\dagger}=I_{(X-W)\times (X-W)}$. We can obtain a pure state decomposition $\{q_t^\prime,\  |\phi_t^\prime\rangle \langle\phi_t^\prime|\}_{t=1,\cdots,Y}$ for $\rho_1$ with $\sqrt{q_t^\prime}|\phi_t^\prime\rangle=\sum_{s=1}^W U_{st} \sqrt{p_s}|\psi_s\rangle$, and a pure state decomposition $\{q_t^{\prime\prime},\  |\phi_t^{\prime\prime}\rangle\langle \phi_t^{\prime\prime}|\}_{t=1,\cdots,Y}$ for $\rho_2$ with $\sqrt{q_t^{\prime\prime}}|\phi_t^{\prime\prime}\rangle=\sum_{s=W+1}^X U_{st} \sqrt{p_s}|\psi_s\rangle$. Since $\{p_s,\  |\psi_s\rangle \langle\psi_s| \}_{s=1,\cdots,W}$ and $\{p_s,\  |\psi_s\rangle \langle\psi_s|\}_{s=W+1,\cdots,X}$ are optimal pure state decompositions attaining the minimum of average $l_1$-norm coherence of $\rho_1$ and $\rho_2$, respectively, we have
\begin{equation}
\begin{array}{rcl}\label{th eq cl1-c >}
\sum_{t=1}^Y q_t C_{l_1}(|\phi_t\rangle \langle\phi_t| )
&\geq&\sum_{t=1}^Y C_{l_1}(\sum_{s=1}^W \sum_{s^\prime=1}^W U_{st}U_{s^\prime t}^\dagger \sqrt{p_s}  \sqrt{p_{s^\prime}} |\psi_s\rangle \langle \psi_{s^\prime}|)\\&&+\sum_{t=1}^Y C_{l_1}(\sum_{s=W+1}^X \sum_{s^\prime=W+1}^X U_{st}U_{s^\prime t}^\dagger \sqrt{p_s}  \sqrt{p_{s^\prime}} |\psi_s\rangle \langle \psi_{s^\prime}|)\\
&=&\sum_{t=1}^Y q_t^\prime C_{l_1}(|\phi_t^\prime\rangle \langle \phi_t^\prime|)+\sum_{t=1}^Y q_t^{\prime\prime} C_{l_1}(|\phi_t^{\prime\prime}\rangle \langle \phi_t^{\prime\prime}|)\\
&\geq& \sum_{s=1}^W p_s C_{l_1}(|\psi_s\rangle \langle\psi_s|) + \sum_{s=W+1}^X p_s C_{l_1}(|\psi_s\rangle\langle\psi_s|)\\
&=& \sigma_1C_{l_1}^{c}(\rho_1)+\sigma_2C_{l_1}^{c}(\rho_2).
\end{array}
\end{equation}
Combining Eqs. (\ref{th eq cl1-c <}) and (\ref{th eq cl1-c >}), we obtain the relation $C_{l_1}^{c}(\rho)=\sigma_1C_{l_1}^{c}(\rho_1)+\sigma_2C_{l_1}^{c}(\rho_2)$. Furthermore, the union of the optimal pure state decompositions of $\sigma_1\rho_1$ and $\sigma_2\rho_2$ is the optimal pure state decomposition of $\rho$.
The general result in Eq. (\ref{theorem1}) can be shown in an analogous manner.
\end{proof}

Now we are ready to calculate the coherence concurrence for X states.
\begin{definition}
The $n$ dimensional X states are quantum states with density matrices in X
shape,
\begin{eqnarray}\label{form of n-dim X state}
\rho=
\left(\begin{array}{cccccccc}
\rho_{11} &0 &0  &\cdots &0  &0  &  \rho_{1,n}\\
 0&\rho_{22} &0  &\cdots &0  &\rho_{{2},{n-1}}  &  0\\
  0& 0 &\rho_{33} &\cdots &\rho_{{3},{n-2}}  & 0 &  0\\
 \cdots& \cdots &\cdots &\cdots &\cdots  & \cdots & \cdots\\
 0& 0 &\rho_{{n-2},3} &\cdots &\rho_{{n-2},{n-2}}  & 0 &  0\\
 0&\rho_{{n-1},2} &0  &\cdots &0  &\rho_{{n-1},{n-1}}  &  0\\
 \rho_{n,1} &0 &0  &\cdots &0  &0  &  \rho_{nn}
\end{array}\right).
\end{eqnarray}
\end{definition}
Without loss of generality, we suppose $\rho_{ii}\leq  \rho_{n-i,n-i}$ for X state, $1\leq i \leq [n/2]$.

\begin{theorem}\label{th X coherence concurrence}
The coherence concurrence of the $n$ dimensional X state $\rho$ given in (\ref{form of n-dim X state}),  is
$$C_{l_1}^{c}(\rho)=2\sum_{i=1}^{[n/2]}|\rho_{i,n+1-i}|.$$
If $n$ is even, an optimal decomposition of $\rho$ is $\{p_i, |\psi_i\rangle\langle\psi_i|\}_{i=1}^n$ with $|\psi_i\rangle\langle\psi_i|=\frac{1}{p_i}\left(\begin{array}{cccccccc}
\rho_{i,i} &\rho_{i,n+1-i}\\
\rho_{i,n+1-i}^*& |\rho_{{i,n+1-i}}|^2/\rho_{i,i}
\end{array}\right)$ with $p_i=\rho_{i,i}+|\rho_{i,n+1-i}|^2/\rho_{i,i}$ for $1\leq i \leq [n/2]$; and $|\psi_i\rangle\langle\psi_i|=|i\rangle \langle i|$ with $p_i=\rho_{n+1-i,n+1-i}-|\rho_{i,n+1-i}|^2/\rho_{i,i}$ for $[n/2]+1\leq i \leq n$.
If $n$ is odd, an optimal decomposition of $\rho$ is $\{p_i, |\psi_i\rangle\langle\psi_i|\}_{i=1}^n$ with $|\psi_i\rangle\langle\psi_i|=\frac{1}{p_i}\left(\begin{array}{cccccccc}
\rho_{i,i} &\rho_{i,n+1-i}\\
\rho_{i,n+1-i}^*& |\rho_{{i,n+1-i}}|^2/\rho_{i,i}
\end{array}\right)$ with $p_i=\rho_{i,i}+|\rho_{i,n+1-i}|^2/\rho_{i,i}$ for $1\leq i \leq [n/2]$; $|\psi_i\rangle\langle\psi_i|=|i\rangle \langle i|$ with $p_i=\rho_{n+1-i,n+1-i}-|\rho_{i,n+1-i}|^2/\rho_{i,i}$ for $[n/2]+1< i \leq n$; $|\psi_{[n/2]+1}\rangle\langle\psi_{[n/2]+1}|=|[n/2]+1\rangle \langle [n/2]+1|$ with probability $p_{[n/2]+1}=\rho_{[n/2]+1,[n/2]+1}$.
\end{theorem}

\begin{proof}
First, if $n$ is even, the X state can be decomposed as the direct sum of $n/2$ quantum states $\rho=\rho_1\oplus\rho_2\oplus\cdots\rho_{n/2}$ up to some permutations, with $\rho_1=
\left(\begin{array}{cccccccc}
\rho_{11} &\rho_{1,n}\\
\rho_{1,n}^*&\rho_{nn}
\end{array}\right)$,
$\rho_2=
\left(\begin{array}{cccccccc}
\rho_{22} &\rho_{2,n-1}\\
\rho_{n-1,2}^*&\rho_{n-1,n-1}
\end{array}\right)$,
$\cdots$,
$\rho_{n/2}=
\left(\begin{array}{cccccccc}
\rho_{n/2,n/2+1} &\rho_{n/2,n/2+1}\\
\rho_{n/2,n/2+1}^*&\rho_{n/2+1,n/2+1}
\end{array}\right)$.
If $n$ is odd, the X state can be decomposed as the direct sum of $[n/2]$ quantum states plus an additional one dimensional matrix up to some permutations, with $[n/2]$ denoting the integer part of the
number $n/2$. The permutations of the matrices neither change the $l_1$ norm coherence nor the coherence concurrence. In any case, the coherence concurrence of an X state $\rho$ is the sum of the coherence concurrence of $\rho_1$, $\rho_2$, $\cdots$, $\rho_{[n/2]}$ by Theorem \ref{th direct sum}. Since $\rho_i$ is a two dimensional state and its coherence concurrence is $C_{l_1}^{c}(\rho_i)=2|\rho_{i,n+1-i}|$ for $i=1,\cdots, [n/2]$, then it is obvious that $C_{l_1}^{c}(\rho)=2\sum_{i=1}^{[n/2]}|\rho_{i,n+1-i}|$. The optimal pure state decomposition follows from the two dimensional case.
\end{proof}

\section{Relation Between Coherence Concurrence and Entanglement}

The coherence of a quantum state $\rho=\sum \rho_{ij}|i\rangle \langle j|$ is closely related to the entanglement of the Schmidt correlated state $\rho_{mc}=\sum \rho_{ij}|ii\rangle \langle jj|$ \cite{E. Rains,A. Winter}.
For example, the coincidence of coherent cost and coherence of formation is identified with the coincidence of entanglement cost and entanglement of formation \cite{A. Winter}.
The relative entropy of coherence of assistance of $\rho$ is equal to the entanglement of assistance of $\rho_{mc}$ \cite{E. Chitambar}.
Here we focus on the entanglement measure called negativity for bipartite quantum states \cite{G. Vidal} and build a relation between coherence concurrence and entanglement.

\begin{definition}
The negativity of quantum state $\rho$ is
$N(\rho)= \frac {\|
\rho^{T_1}\|-1}{2}$,
which corresponds to the absolute value of the
sum of negative eigenvalues of $\rho^{PT}$, the superscript $PT$ means the partial transposition.
\end{definition}

Based on negativity, the convex roof extended negativity is proposed by Ref. \cite{S. Lee}, which is also an entanglement measure.
\begin{definition}
The convex roof extended negativity $N_c(\rho)$ is defined as $N_c(\rho)=\min \sum_i p_i N(|\psi_i\rangle\langle\psi_i|)$,
where the minimization is taken over all pure state decompositions of $\rho=\sum_i p_i |\psi_i\rangle\langle\psi_i|$.
\end{definition}
The $l_1$-norm coherence itself corresponds to the negativity by ${C}_{l_1}(\rho)=2N(\rho_{mc})$ \cite{S. Rana,H. Zhu2017}.
Next we show the relation between the coherence concurrence and the convex roof extended negativity.

\begin{theorem}
The coherence concurrence of $\rho=\sum \rho_{ij}|i\rangle \langle j|$ is twice the convex roof extended negativity of $\rho_{mc}$,
\begin{equation}
C_{l_1}^{c}(\rho)=2N_c(\rho_{mc}).
\end{equation}
\end{theorem}

\begin{proof}
Note that for the maximally correlated state $\rho_{mc}$, the pure state decompositions are all in Schmidt form $|\psi\rangle=\sum_i a_i|ii\rangle$ \cite{M. J. Zhao2008}. For pure state $|\psi\rangle=\sum_i a_i|ii\rangle$, the $l_1$-norm coherence is $C_{l_1}(|\psi\rangle\langle\psi|)=\sum_{i\neq j} |a_i^*a_j|$ and the negativity is $N(|\psi\rangle\langle\psi|)=\frac{1}{2}\sum_{i\neq j} |a_i^*a_j|$ by definitions. Hence the $l_1$-norm coherence is twice the negativity, $2N(|\psi\rangle\langle\psi|)=C_{l_1}(|\psi\rangle\langle\psi|)$.
Therefore, if $\{ p_k,\  |\psi^{\prime}_k\rangle\langle \psi^{\prime}_k|\}$ is the optimal pure state decomposition for $\rho_{mc}$ such that $N_c(\rho_{mc})=\sum_k p_k N(|\psi^{\prime}_k\rangle\langle \psi^{\prime}_k|)$ with $|\psi^{\prime}_k\rangle=\sum_i a_i^{(k)}|ii\rangle$, then $\{ p_k,\  |\psi_k\rangle\langle\psi_k|\}$ with $|\psi_k\rangle=\sum_i a_i^{(k)}|i\rangle$ is the optimal pure state decomposition for $\rho$ such that $C_{l_1}^{c}(\rho)=\sum_k p_k C_{l_1}(|\psi_k\rangle\langle\psi_k|)$.
\end{proof}


\section{Conclusions}

We have studied the properties of coherence concurrence. Analogous to the coherence of assistance we have given a physical explanation for coherence concurrence. The optimal pure state decomposition attaining the coherence concurrence has been presented for qubit states. The additivity of coherence concurrence under the direct sum operation has been proved alternatively. Since the X state is the direct sum of qubit states, the coherence concurrence for the X states has been proved to be equal to the $l_1$-norm coherence and the optimal pure state decompositions are provided. Moreover, it
has been shown that the coherence concurrence is just twice the convex roof extended negativity
of the Schmidt correlated states. Originating from the superposition principle in quantum mechanics,
coherence is a fundamental phenomena of quantum world. Our results may highlight further investigations
on quantum coherence, for example, the relations among coherence concurrence, the $l_1$ norm coherence of assistance and the $l_1$ norm coherence.

\bigskip
\noindent{\bf Acknowledgments}\,
Ming-Jing Zhao thanks the Department of Mathematics and Statistics, University of Guelph, Canada for hospitality. This work is supported by the NSF of China (Grant Nos. 11401032, 11501247, 11675113), the China Scholarship Council (Grant No. 201808110022), Qin Xin Talents Cultivation Program, Beijing Information Science and Technology University, Key Project of Beijing Municipal Commission of Education (KZ201810028042), Beijing Natural Science Foundation (Z190005), Open Foundation of State Key Laboratory of Networking and Switching Technology (Beijing University of Posts and
Telecommunications) (SKLNST-20**-2-0*) and the National Science and Engineering Research Council of Canada (Discovery Grant No. 400550).

\end{document}